\documentclass{amsart}
\usepackage{amsmath}
\usepackage{amssymb}
\usepackage{graphicx}
\usepackage{amscd}

\newtheorem{theorem}{Theorem}
\theoremstyle{plain}

\newtheorem{lemma}{Lemma}

\newtheorem{remark}{Remark}

\numberwithin{equation}{section}

\begin{document}
\title[Estimates of GPI gap of decomposability]{Statistical estimation of
gap of decomposability of the general poverty index}

\author{Mohamed Cheikh Haidara $^{*}$}
\email{chaidara@ufrsat.org}
\address{LERSTAD, Universit\'{e} de Saint-Louis $^{*}$}
\urladdr{www.statpas.net/cva.php?email=chheikhh@yahoo.fr}

\author{ Gane Samb LO $^{**}$}
\email{ganesamblo@ufrsat.org}
\address{$^{**}$ Laboratoire de Statistiques Th\'eoriques et Appliqu\'ees (LSTA) \\Universit\'e Pierre et Marie Curie (UPMC) France}
\address{LERSTAD, Universit\'{e} de Saint-Louis.}
\urladdr{www.lsta.upmc.fr/gslo}

\begin{abstract}

For the decomposability property is very a practical one in Welfare
analysis, most researchers and users favor decomposable poverty indices such
as the Foster-Greer-Thorbeck poverty index. This may lead to neglect the so
important weighted indices like the Kakwani and Shorrocks ones which have
interesting other properties in Welfare analysis. To face up to this
problem, we give in this paper, statistical estimations of the gap of
decomposability of a large class of such indices using the General Poverty
Indice (GPI) and of a new asymptotic representation Theorem for it, in terms
of functional empirical processes theory. The results then enable
independent handling of targeted groups and next global reporting with
significant confidence intervals. Data-driven examples are given with real
data.
\end{abstract}

\maketitle


\large

\section{Introduction}

\noindent We are concerned in this paper with the statistical estimation of
the gap of decomposability of the class of the statistical poverty indices
in general. Suppose that we have some statistic of the functional form $%
J_{n}=J(Y_{1},...,Y_{n})$ where $\mathcal{E}=\{Y_{1},...,Y_{n}\}$ is a
sample of the random variable $Y$ defined on a probability space $%
(\Omega ,\mathcal{A},\mathbb{P}$) and drawn from some specific population.
Now, suppose that this population is divided into K subgroups $%
S_{1},...,S_{K}$ and let us, for each $i\in \{1,...,K\}$, denote the subset
of the random sample $\{Y_{1},...,Y_{n}\}$ coming from $S_{i}$ by $\mathcal{E%
}_{i}=\{Y_{1,i},...,Y_{n_{i},i}\}$ and then put $%
J_{n_{i}}(i)=J(Y_{1,i},...,Y_{n_{i},i}).$ The statistic $J_{n}$ is said to
be decomposable whenever one always has%
\begin{equation*}
J_{n}=\frac{1}{n}\sum_{i=1}^{K}n_{i}J_{n_{i}}(i),
\end{equation*}%
whatever may be the way in which $\mathcal{E}$ is partitioned into the $%
\mathcal{E}_{i}$ '$s$ $(i=1,...,K).$ This property is a very practical one
when dealing with the poverty measures or welfare measures in general for
the following reason. If we are willing to monitor the poverty situation, it
may be very useful to target some sensitive areas or subgroups. By dividing the
population into targeted groups, and estimating the poverty intensity by $%
J_{n_{i}}(i)$ (resp. variation of poverty by $\Delta J_{n_{i}}(i))$ in each
group, one would be able to report the poverty intensity (resp. global
poverty variation) by (\ref{decomp1}) (resp. $\Delta J_{n}=\frac{1}{n}%
\sum_{i=1}^{K}n_{i}\Delta J_{n_{i}}(i)$), provided that the samples are the
same as it is the case in longitudinal data$.$ Thus, decomposability allows
an independent handling of poverty for different areas and next an easy
reconstruction of the global situation. \bigskip

\noindent Now in the specific case of poverty indices, we mainly have the
non-weighted ones and the weighted ones. The statistics in the first case
are automatically decomposable and then are mostly preferred by users.
However, the weighted measures, which in general are not decomposable, have
very interesting properties in poverty analysis. Dismissing them only for
non-decomposability would result in a disaster. We tackle this problem in
this paper. Indeed, by estimating the following gap of decomposability%
\begin{equation}
gd_{n}=J_{n}-\frac{1}{n}\sum_{i=1}^{K}n_{i}J_{n_{i}}(i)  \label{decomp1}
\end{equation}%
with significant confidence intervals, we would be able to handle separated
analyses in the subgroups and report the global case and, at the same time,
make benefit of the other properties of such statistics.\bigskip

\noindent The remainder of the paper is organized as follows. In Section \ref%
{sec2}, we give a brief introduction of the poverty measures and to the
General Poverty Index (GPI). In Section \ref{sec3}, we return back to the
decomposability problem by describing the drawing scheme under which the
results are given. In Section \ref{sec4}, we state the results which are
applied to the Senegalese and Mauritanian data in Section \ref{sec5}. The
proofs are given in Section \ref{sec6}. The concluding remarks are in
Section \ref{sec7}. The paper is finished by an appendix in Section \ref{sec8}.\\

\section{A brief reminder on Poverty measures}

\label{sec2}

\noindent We consider a population of individuals or households, each of
which having a random income or expenditure $Y$\ with distribution function $%
G(y)=\mathbb{P}(Y\leq y).$\ In the sequel, we use $Y$ as an income variable
although it might be any positive random variable. An individual is
classified as poor\ whenever his income or expenditure $Y$\ fulfills $Y<Z,$\
where $Z$\ is a specified threshold level (the poverty line). \newline

\noindent Consider also a random sample $Y_{1},Y_{2},...Y_{n}$\ of size $n$\
of incomes, with empirical distribution function $G_{n}(y)=n^{-1}\#\left\{
Y_{i}\leq y:1\leq i\leq n\right\} $. The number of poor individuals within
the sample is then equal to $Q_{n}=nG_{n}(Z)$. And, from now on, all the
random elements used in the paper are defined on the same probability space $%
(\Omega ,\mathcal{A},\mathbb{P})$.\\

\noindent Given these preliminaries, we introduce measurable functions $%
A(p,q,z)$, $w(t)$, and $d(t)$\ of $p,q\in \mathbb{N},$\ and $z,t\in \mathbb{R%
}$. Set $B(Q_{n})=\sum_{i=1}^{Q_{n}}w(i).$ \newline

\noindent Let $Y_{1,n}\leq Y_{2,n}\leq ...\leq Y_{n,n}$\ be the order
statistics of the sample $Y_{1},Y_{2},...Y_{n}$\ of $Y$. We consider general
poverty indices (GPI)\ of the form 
\begin{equation}
GPI_{n}=\delta \left( \frac{A(Q_{n},n,Z)}{nB(Q_{n},n)}\sum_{j=1}^{Q_{n}}w(%
\mu _{1}n+\mu _{2}Q_{n}-\mu _{3}j+\mu _{4})\text{\ }d\left( \frac{Z-Y_{j,n}}{%
Z}\right) \right) ,  \label{ssl01}
\end{equation}%
where $\mu _{1},\mu _{2},\mu _{3},\mu _{4}$\ are constants. This global form
of poverty indices was introduced in \cite{lo2} (see also \cite{lo}, \cite%
{lo2} and \cite{sl}) as an attempt to unify the large number of poverty
indices that have been introduced in the literature since the pioneering
work of the Nobel Prize winner, Amartya Sen(1976) who first derived poverty
measures (see \cite{sen}) from an axiomatic point of view. A survey of these
indices is to be found in Zheng \cite{zheng}, who also discussed their
introduction, from an axiomatic point of view. We will cite a few number of
them here just to make clear the minds and prepare the data-driven
applications in Section \ref{sec5}.\bigskip

\noindent One may devide the poverty indices into two classes. The first
includes the nonweighted ones. The most popular of them is the
Foster-Greer-Thorbecke(1984) \cite{fgt} class which is defined for $\alpha
\geq 0$, by 
\begin{equation}
FGT(\alpha )=\frac{1}{n}\overset{Q_{n}}{\underset{j=1}{\sum }}\left( \frac{%
Z-Y_{j,n}}{Z}\right) ^{\alpha }.  \label{ex0}
\end{equation}%
For $\alpha =0,$ (\ref{ex0}) reduces to $Q_{n}/n,$\ the headcount of poor\
individuals. For $\alpha =1$\ and $\alpha =2,$\ it is respectively
interpreted as the severity of poverty and the depth in poverty. (\ref{ex0})
is obtained from (\ref{ssl01}) by taking 
\begin{equation*}
\begin{array}{cccc}
\delta =I_{d}, & w\equiv 1, & d(u)=u^{\alpha }, & B(Q_{n},n)=Q_{n}\text{ and 
}A(Q_{n},n,Z)=Q_{n}.%
\end{array}%
\end{equation*}%
\noindent Next, we have for $\alpha \geq 0$, 
\begin{equation*}
C(\alpha )=\frac{1}{n}\overset{Q_{n}}{\underset{j=1}{\sum }}\left( 1-\left( 
\frac{Y_{j,n}}{Z}\right) ^{\alpha }\right) ,
\end{equation*}%
\noindent the Chakravarty family class of poverty measures is obtained from (%
\ref{ssl01}) by taking $Y^{\alpha }$ and $Z^{\alpha }$ as respectively
transformed income $Y$\ and threshold $Z$ and 
\begin{equation*}
\begin{array}{cccc}
\delta =I_{d}, & w\equiv 1, & d(u)=u, & B(Q_{n},n)=Q_{n}\text{ and }%
A(Q_{n},n,Z)=Q_{n}.%
\end{array}%
\end{equation*}%
The statistics in this class are decomposable and are not concerned by the
present work. \bigskip

\noindent The second class consists of the weighted indices. We mention here
two of its famous members. The Sen(1976) index (see \cite{sen}) 
\begin{equation}
P_{SE,n}=\frac{2}{n(Q_{n}+1)}\overset{Q_{n}}{\underset{j=1}{\sum }}%
(Q_{n}-j+1)\left( \frac{Z-Y_{j,n}}{Z}\right)  \label{ex2}
\end{equation}%
(\ref{ex2}) is obtained from (\ref{ssl01}), by taking 
\begin{equation*}
d(u)=u,\text{ }w(u)=u,\text{ }A(Q_{n},n,Z)=Q_{n},
\end{equation*}%
\begin{equation*}
B(Q_{n})=Q_{n}(Q_{n}+1)/2,\text{ }\mu _{1}=0\text{ and }\mu _{3}=\mu
_{2}=\mu _{4}=1.
\end{equation*}%
The Shorrocks(1995) index (see \cite{shorrocks}) 
\begin{equation}
P_{SH,n}=\frac{1}{n^{2}}\overset{Q_{n}}{\underset{j=1}{\sum }}%
(2n-2j+1)\left( \frac{Z-Y_{j,n}}{Z}\right)  \label{ex3}
\end{equation}%

\noindent is obtained from (\ref{ssl01}) by taking

\begin{equation*}
B(Q_{n},n)=Q_{n}(Q_{n}+1)/2,\text{ }A(n,Q_{n},Z)=Q_{n}(Q_{n}+1)/2n,
\end{equation*}%
\ 
\begin{equation*}
\begin{array}{ccccccc}
\delta =I_{d}, & w(u)\equiv (u), & d(u)=u, & \mu _{1}=2, & \mu _{2}=0, & \mu
_{3}=2 & \mu _{4}=1.%
\end{array}%
\end{equation*}%
Measures (\ref{ex2}) and (\ref{ex3}) evaluate the poverty intensity by giving a more
important weight on the poorest individuals. This means that a small
decrease of the intensity on the poorest household indicates significant
improvement in the population.\\

\noindent In the applications, we mainly deal with these two specific
measures because of their importance in poverty analysis. Notice that the
Thon measure (\cite{thon}) is different from the Shorrocks one only by their normalization
coefficients which are respectively $n(n+1)$ and $n^{2},$ so that they have
the same asymptotic behavior. Finally, we have the following generalization
of the Sen measure given by Kakwani(1980) \cite{kakwani},

\begin{equation*}
J_{n}(k)=\frac{Q}{n\sum_{j=1}^{Q}j^{k}}\sum_{j=1}^{Q}(Q-j+1)^{k}d\left( 
\frac{Z-Y_{j,n}}{Z}\right) ,
\end{equation*}

\noindent where $k$ is a positive parameter. Notice that $J_{n}(1)$ is the
Sen measure. Notice also that, under mild conditions, $J_{n}$ converges in
probability to the Exact General Poverty Index (EGPI) (see \cite{barrett}, 
\cite{bchowzheng}, \cite{bfzheng} and \cite{lo}), 
\begin{equation}
J(G)=GPI=\int_{0}^{Z}L_{1}(u,G)\text{ }d\left( \frac{Z-u)}{Z}\right) dG(u),
\label{ssl010}
\end{equation}%
where $L_{1}$ is some weight function depending on the distribution
function. This result will be proved again in Theorem \ref{theo1} below.

\section{Statistical decomposability}

\label{sec3}

\noindent From now, we suppose that our studied population of households is
divided into $K$ subgroup such that, for each $i\in \{1,...,K\}$, the
probability that a randomly drawn household comes from the $i^{th}$ subgroup
is $p_{i}>0$, with $p_{1}+...+p_{K}=1$. Let us suppose that we draw a sample
of size n from the population : $Y_{1},...,Y_{n}$ and let us denote those of
the $n_{i}^{\ast }$ observations coming from the $i^{th}$ subgroup, $(1\leq
i\leq K)$ by $Y_{i,j},$ $j=1,...,n_{i}^{\ast }.$ Let $J_{n_{i}^{\ast
}}(G_{i})=J_{n_{i}^{\ast }}(Y_{i,1},...,Y_{i,n_{i}^{\ast }})$ the empirical
index measured on the $i^{th}$ subgroup and $J_{n}(G)$ the global index.
Clearly, decomposability implies for all $n\geq 1$,%
\begin{equation*}
gd_{n}=J_{n}-\frac{1}{n}\sum_{i=1}^{K}n_{i}^{\ast }J_{n_{i}^{\ast }}\equiv 0.
\end{equation*}%
Surely, $n^{\ast }=(n_{1}^{\ast },...n_{K}^{\ast })$ follows a multinomial
law with parameters $n$ and $p=(p_{1},...,p_{K}).$ Since each $p_{i}>0,$ we
have that for each $1\leq i \leq K$, $n_{i}^{\ast }\rightarrow \infty $ $a.s.$, as $n\rightarrow
\infty $. We will have by (\ref{decomp1}) and by (\ref{ssl010}),%
\begin{equation*}
gd_{n}=J_{n}(G)-\frac{1}{n}\sum_{i=1}^{K}n_{i}^{\ast }J_{n_{i}^{\ast
}}(G_{i})\rightarrow _{P}gd=J(G)-\sum_{i=1}^{K}p_{i}J_{i}(G_{i}).
\end{equation*}%
The right member of this equation is the exact gap of decomposability $gd$.
It follows that $gd$ is zero if the distribution of the income is the same
over all the population, that the more homogeneous the income is over the
population, the lower the gap of decomposability $gd$ is. As a first result,
we get that the decomposability does not, asymptotically at least, matter
for a more or less homogeneous population. That is, the decomposability is
not only a functional form matter (of the index), but it is also a
statistical one since whatever might be the index, decomposability is
asymptotically obtained when the subgroups have the same distribution. For
example, it has been pointed out in (\cite{halo}), for the Senegalese
poverty databases from 1996 to 2001, that the gaps of decomposability were
very low for various stratifications (in regions, gender, ethnic groups,
etc.). The apparent reason was the homogeneity of the income. Such results
are confirmed in Section \ref{sec5}.\bigskip

\noindent Now we want to find the law of 
\begin{equation*}
gd_{n}^{\ast }=\sqrt{n}(gd_{n}-gd)
\end{equation*}%
for a more accurate estimation of $gd$ by confidence intervals. At this
step, we have to precise our random scheme. We put a probability space ($%
\Omega _{1}\times \Omega _{2},\mathcal{A}_{1}\otimes \mathcal{A}_{2},\mathbb{%
P}_{1}\otimes \mathbb{P}_{2})$ and put $\mathbb{P=P}_{1}\otimes \mathbb{P}%
_{2}.$ We draw the observations in the following way. In each trial, we draw
a subgroup, the $ith$ subgroup $(\mathcal{E}_{i})$\ having the occurring
probability $p_{i}.$\ And we put 
\begin{equation*}
\pi _{i,j}(\omega _{1})=\mathbb{I}_{(\text{the }i^{th}\text{ subgroup is
drawn at the }j^{th}\text{ trial})}(\omega _{1}),
\end{equation*}%
\noindent $1\leq i\leq K,1\leq j\leq n$. Now, given that the $i^{th}$
subgroup is drawn at the $j^{th}$ trial, we pick one individual in this
subgroup and observe its income $Y_{j}(\omega _{1},\omega _{2}).$ We then
have the observations 
\begin{equation*}
\{Y_{j}(\omega _{1},\omega _{2}),\text{ }1\leq j\leq n\}.
\end{equation*}%
We have these simple facts. First, \ for $1\leq i\leq K,$ 
\begin{equation}
n_{i}^{\ast }=\sum_{j=1}^{n}\pi _{i,j}.
\end{equation}%
Secondly, the distribution of $Y_{j}$ given $(\pi _{i,j}=1)$, is $G_{i}$,
that is 
\begin{equation*}
\mathbb{P}(Y_{j}\leq y\text{ }\diagup \pi _{i,j}=1)=G_{i}(y),
\end{equation*}%
Next 
\begin{eqnarray*}
\forall (y &\in &\mathbb{R)}, \\
\mathbb{P}(Y_{j} &\leq &y\text{ })=\sum_{i=1}^{K}\mathbb{P}(\pi _{i,j}=1)%
\mathbb{P}(Y_{j}\leq y\text{ }\diagup \pi
_{i,j}=1)=\sum_{i=1}^{K}p_{i}G_{i}(y).
\end{eqnarray*}%
We conclude that $\{Y_{1},...,Y_{n}\}$ is an independent sample drawn from $%
G(y)$ $=\sum_{i=1}^{K}p_{i}G_{i}(y),$ the mixture of the distribution
functions of the subgroups incomes. Finally, we readily see that
conditionally on $n^{\ast }\equiv (n_{1}^{\ast },n_{2}^{\ast
},...,n_{K}^{\ast })=(n_{1},n_{2},...,n_{K})\equiv \overline{n}$ with $%
n_{1}+n_{2}+...+n_{K}=n,$ $\{Y_{i,j},$ $1\leq j\leq n_{i}^{\ast }\}$ are
independent random variables with distribution function $G_{i}$.

\section{Our results}

\label{sec4}

\noindent The results stated here hold for a very large class of poverty
measures summarized in the GPI. This is why we need the representation
Theorem of the GPI in \cite{slrep}. In fact, we do not need here the
complete form of \cite{slrep}, but a special case of it, based on the
assumptions described below. For that, suppose that $G_{i}$ $(1\leq i\leq K)$%
, \ is the distribution function of the income for the $ith$ subgroup, and $%
G $ is the distribution function of the income for the global population.
Let also $\gamma (x)=d\left( \frac{Z-x}{Z}\right) 1_{(x\leq Z)}$ and $%
e(x)=1_{(x\leq Z)}$. The following assumptions are required.

\begin{itemize}
\item[(HD0)] $G_{0}(Z)\in ]0,1[$ for $G_{0}\in \{G,G_{1},...,G_{K}\}.$

\item[(HD1)] There exist a function $h(p,q)$ of $(p,q)\in \mathbb{N}^{2}$
and a function $c(s,t)$ of $(s,t)\in (0,1)^{2}$ such that, as $n\rightarrow
+\infty $, 
\begin{multline*}
\max_{1\leq j\leq Q}\left\vert A(n,Q)h^{-1}(n,Q)w(\mu _{1}n+\mu _{2}Q-\mu
_{3}j+\mu _{4})-c(Q/n,j/n)\right\vert \\
=o_{P}(n^{-1/2}).
\end{multline*}

\item[(HD2)] For the function $h$ found in $(HD1)$, there exists a function $\pi (s,t)$ of $(s,t)\in \mathbb{R}^{2}$
such that as $n\rightarrow +\infty $, 
\begin{equation*}
\max_{1\leq j\leq Q}\left\vert w(j)h^{-1}(n,Q)-\frac{1}{n}\pi
(Q/n,j/n)\right\vert =o_{P}(n^{-3/2}).
\end{equation*}

\item[(HD3)] The bivariate functions $c$ and $\pi $ have continuous partial
differentials.

\item[(HD4)] For a fixed $x$, the functions $y\rightarrow \frac{\partial c}{%
\partial y}(x,y)$ and $y\rightarrow \frac{\partial \pi }{\partial y}(x,y)$
are monotone.

\item[(HD5)] $G_{0}$ is strictly increasing for any $G_{0}\in
\{G,G_{1},...,G_{K}\}$.

\item[(HD6)] We have for any $G_{0}\in \{G,G_{1},...,G_{K}\}$\newline
\begin{equation*}
0<H_{c}(G_{0})=\int c(G_{0}(Z),G_{0}(y))\gamma (y)dG_{0}(y)<+\infty
\end{equation*}%
and 
\begin{equation*}
0<H_{\pi }(G_{0})=\int \pi (G_{0}(Z),G_{0}(y))e(y)dG_{0}(y)<+\infty
\end{equation*}%

\noindent We also need the following definitions, for $G_{0}\in \{G,G_{1},...,G_{K}\},$%
\newline
\begin{equation*}
J(G_{0})=H_{c}(G_{0})/H_{\pi }(G_{0}),
\end{equation*}%
\begin{equation}
g_{0}(\cdot )=H_{\pi }^{-1}(G_{0})g_{c,0}(\cdot )-H_{c}(G_{0})H_{\pi
}^{-2}(G_{0})g_{\pi ,0}(\cdot )+K(G_{0})e(\cdot ),  \label{g0}
\end{equation}%

\noindent with
 
\begin{equation}
g_{c,0}(\cdot )=c(G_{0}(Z),G_{0}(\cdot ))\gamma (\cdot ),\text{ }g_{\pi
,0}(\cdot )=\pi (G_{0}(Z),G_{0}(\cdot ))e(\cdot ),  \label{meth2}
\end{equation}%
\begin{equation}
K(G_{0})=H_{\pi }^{-1}(G_{0})K_{c}(G_{0})-H_{c}(G_{0})H_{\pi
}^{-2}(G_{0})K_{\pi }(G_{0})  \label{meth3}
\end{equation}%

\noindent with
 
\begin{multline}
K_{c}(G_{0})=\int_{0}^{1}\frac{\partial c}{\partial x}(G_{0}(Z),s)\gamma
(G_{0}^{-1}(s))ds,  \label{meth4} \\
\text{ }K_{\pi }(G_{0})=\int_{0}^{1}\frac{\partial \pi }{\partial x}%
(G_{0}(Z),s)e(G_{0}^{-1}(s))ds,
\end{multline}

\begin{equation}
\nu _{0}(\cdot )=H_{\pi }^{-1}(G_{0})\nu _{c,0}(\cdot )-H_{c}(G_{0})H_{\pi
}^{-2}(G_{0})\nu _{\pi ,0}(\cdot ),  \label{nu0}
\end{equation}

\noindent where 

\begin{equation*}
\nu _{c,0}(\cdot )=\frac{\partial c}{\partial y}(G_{0}(Z),G_{0}(\cdot
))\gamma (\cdot ),\nu _{\pi ,0}(\cdot )=\frac{\partial \pi }{\partial y}%
(G_{0}(Z),G_{0}(\cdot ))e(\cdot ).
\end{equation*}
\end{itemize}

\noindent with the conventions that for $G_{0}=G,$ we denote $g_{0}=g$ and $%
\nu _{0}=\nu .$ For \ $G_{0}=G_{i},1\leq i\leq K,$ we put $g_{0}=g_{i}$ and $%
\nu _{0}=\nu _{i}$. Finally define 
\begin{equation}
\ell_{i}(t)=(g-g_{i})\left( G_{i}^{-1}(t)\right) ,\text{ }c_{i}(t)=(p_{i}\nu
-\nu _{i})\left( G_{i}^{-1}(t)\right) ,\text{ }0\leq t\leq 1. \label{def0}
\end{equation}
\noindent We are now able to briefly describe the approximation of \cite%
{slrep}\ : if $G_{0}$ fulfills (HD1), ..., (HD6), then as $n\rightarrow
+\infty ,$ we have%
\begin{equation*}
\sqrt{n}(J_{n}(G_{0})-J(G_{0}))=\alpha _{n}(g_{0})+\beta _{n}(\nu
_{0})+o_{P}(1),
\end{equation*}

\noindent where

\begin{equation*}
\alpha _{n}(g_{0})=\frac{1}{\sqrt{n}}\sum_{j=1}^{n}g_{0}(G_{0}(V_{j})-%
\mathbb{E}g_{0}(G_{0}(V_{j}))
\end{equation*}%

\noindent is the functional empirical process and 
\begin{equation}
\beta _{n}(\nu _{0})=\frac{1}{\sqrt{n}}\sum_{j=1}^{n}\left\{
G_{n}(V_{j})-G_{0}(V_{j})\right\} \nu _{0}(V_{j})
\end{equation}%
\noindent is a \textit{residual} stochastic process introduced in \cite{slrep} and widely studied in \cite{losto}, where $G_{n}$ is the empirical
distribution function associated with $\{V_{1},...,V_{n}\}$ sampled from $%
G_{0}.$

\bigskip

\noindent Finally, we introduce these constants of whom the variances of our
theorem are based on :

\begin{equation*}
A_{1}=\sum_{i=1}^{K}p_{i}\left\{ \int_{0}^{G_{i}(Z)}(\overline{g}-\overline{g%
}_{i})^{2}(G_{i}^{-1}(t))dt-\left( \int_{0}^{G_{i}(Z)}(\overline{g}-%
\overline{g}_{i})(G_{i}^{-1}(t))dt\right) ^{2}\right\} ,
\end{equation*}%
\begin{equation*}
A_{2}=\sum_{i}^{K}p_{i}\int_{0}^{G_{i}(Z)}\int_{0}^{G_{i}(Z)}(s\wedge
t-st)(p_{i}\overline{\nu }-\overline{\nu }_{i})(G_{i}^{-1}(s))(p_{i}%
\overline{\nu }-\overline{\nu }_{i})(G_{i}^{-1}(t))dsdt,
\end{equation*}%

$$
A_{31} = \sum_{i=1}^{K}p_{i}\sum_{h\neq i}^{K}p_{h}^{2} 
\int_{0}^{G_{i}(Z)}\int_{0}^{G_{i}(Z)} \left[ {G_{h}(G_{i}^{-1}(s))\wedge
G_{h}(G_{i}^{-1}(t))} \right.
$$

$$
\left. {-G_{h}(G_{i}^{-1}(s))G_{h}(G_{i}^{-1}(t))} \right] \\
\overline{\nu }(G_{i}^{-1}(s))\overline{\nu }(G_{i}^{-1}(t))dsdt,
$$

$$
A_{32} =\sum_{i=1}^{K}p_{i}^{1/2}\sum_{j\neq
i}^{K}p_{j}^{1/2}\sum_{h\notin \{i,j\}}^{K}p_{h}^{2} 
\int_{0}^{G_{i}(Z)}\int_{0}^{G_{j}(Z)} \left[{G_{h}(G_{i}^{-1}(s))\wedge G_{h}(G_{j}^{-1}(t))} \right.
$$

$$
\left. {-G_{h}(G_{i}^{-1}(s))G_{h}(G_{j}^{-1}(t))} \right]
 \overline{\nu }(G_{i}^{-1}(s))\overline{\nu }(G_{j}^{-1}(t))dsdt,
$$

$$
B_{1} =\sum_{i=1}^{K}p_{i} \int_{0}^{G_{i}(Z)} \left\{ {\int_{0}^{s\wedge G_{i}(Z)}(\overline{g}-
\overline{g}_{i})(G_{i}^{-1}(t))dt} \right.
$$

$$
\left. {-s\int_{0}^{1}(\overline{g}-\overline{g}%
_{i})(G_{i}^{-1}(t))dt} \right\} (p_{i}\overline{\nu }-\overline{\nu }_{i})(G_{i}^{-1}(s))ds,
$$

$$
B_{2} = \sum_{i=1}^{K}p_{i}^{3/2}\sum_{j\neq i}^{K}p_{j}^{1/2} \int_{0}^{G_{i}(Z)}\int_{0}^{G_{j}(Z)}[s\wedge
G_{i}(G_{j}^{-1}(t))-sG_{i}(G_{j}^{-1}(t))] 
$$
$$
\times (p_{i}\overline{\nu }-\overline{\nu }_{i})(G_{i}^{-1}(s))\overline{%
\nu }(G_{j}^{-1}(t))dsdt,
$$

\noindent and\\

$$
B_{3} =\sum_{i=1}^{K}p_{i}^{3/2}\sum_{j\neq i}^{K}p_{j}^{1/2} \int_{0}^{G_{j}(Z)} \\
\left\{ {\int_{0}^{G_{i}(G_{j}^{-1}(s))\wedge G_{i}(Z)}(\overline{g}-%
\overline{g}_{i})(G_{i}^{-1}(t))dt} \right.
$$

$$
\left. {-G_{i}(G_{j}^{-1}(s))\times \int_{0}^{1}(%
\overline{g}-\overline{g}_{i})(G_{i}^{-1}(t))dt} \right\} 
\overline{\nu }(G_{j}^{-1}(s))ds,
$$

\noindent where%
\begin{equation*}
g_{0}(\cdot )=\overline{g}_{0}(\cdot )\times e(\cdot )\text{ and } \nu _{0}(\cdot )=\overline{\nu }_{0}(\cdot )\times e(\cdot ),
\end{equation*}%
\noindent and

\begin{equation*}
\left( g_{0},\nu _{0}\right) \in \left( g,g_{1},...,g_{K}\right) \times
\left( \nu ,\nu _{1},...,\nu _{K}\right) \text{ and }i=1,...,K.
\end{equation*}%
\bigskip

\noindent We are now able to state our main result. \bigskip

\begin{theorem}
\label{theo1} Let (HD0)-(HD6) hold. Then $gd_{n,0}^{\ast }=\sqrt{n}%
(gd_{n}-gd_{0})\leadsto \mathcal{N}(0,\vartheta _{1}^{2}+\vartheta
_{3}^{2}), $ and $gd_{n}^{\ast }=\sqrt{n}(gd_{n}-gd)\leadsto \mathcal{N}%
(0,\vartheta _{1}^{2}+\vartheta _{2}^{2})$ with
\begin{equation*}
\vartheta _{1}^{2}=A_{1}+A_{2}+A_{3}+2(B_{1}+B_{2}+B_{3})
\end{equation*}%
\begin{equation*}
\vartheta _{2}^{2}=\sum_{h=1}^{K}F_{h}{}^{2}p_{h}-\left(
\sum_{h=1}^{K}F_{h}p_{h}\right) ^{2}
\end{equation*}%
for $F_{h}=\mathbb{E}g(Y^{h})-J(G_{h})+\sum_{i=1}^{K}p_{i}\mathbb{E}%
G_{h}(Y^{i})\nu (Y^{i}),$ and%
\begin{equation*}
\vartheta _{3}^{2}=\sum_{h=1}^{K}M_{h}{}^{2}p_{h}-\left(
\sum_{h=1}^{K}M_{h}p_{h}\right) ^{2}
\end{equation*}%
for $M_{h}=\mathbb{E}g(Y^{h})+\sum_{i=1}^{K}p_{i}\mathbb{E}G_{h}(Y^{i})\nu
(Y^{i}).$
\end{theorem}

\bigskip

\begin{remark}
This clearly makes the so important decomposability requirement less crucial
since the default of decomposability may be estimated by confidence
intervals based on this theorem, as we showed it in the next section.
\end{remark}

\section{Examples and Applications}

\label{sec5}

\subsection{Sen Case}

The conditions (HD1), (HD2), (HD3) and (HD4) hold for this measure and we
have here $c(x,y)=x-y$ and $\pi (x,y)=y/x.$ Further when (HD0), (HD5) and
(HD6) are true,\ the results of Theorem \ref{theo1} apply with%
\begin{equation*}
J(G_{0})=2\int_{0}^{G_{0}(Z)}\left( 1-\frac{s}{G_{0}(Z)}\right) \left( \frac{%
Z-G_{0}^{-1}(s)}{Z}\right) ds,
\end{equation*}%
\begin{equation*}
K(G_{0})=2\left( 1-\frac{1}{ZG_{0}(Z)}\int_{0}^{G_{0}(Z)}G_{0}^{-1}(s)ds%
\right) +\frac{J(G_{0})}{G_{0}(Z)},
\end{equation*}%
\begin{eqnarray*}
g_{0}(y) &=&2\left\{ \left[ \left( 1-\frac{G_{0}(y)}{G_{0}(Z)}\right) \left( 
\frac{Z-y}{Z}\right) \right. \right. \\
&&-\left. \left. \left( \frac{G_{0}(y)}{G_{0}(Z)}\right) \left( \frac{%
J(G_{0})}{G_{0}(Z)}\right) \right] +K(G_{0})\right\} 1_{(y\leq Z)},
\end{eqnarray*}%
and%
\begin{equation*}
\nu _{0}(y)=-\frac{2}{G_{0}(Z)}\left[ \left( \frac{Z-y}{Z}\right) +\frac{%
J(G_{0})}{G_{0}(Z)}\right] 1_{(y\leq Z)}.
\end{equation*}

\subsection{Shorrocks' case}

We have the same conclusion of the previous case with $c(x,y)=2(1-y),$ $%
K(G_{0})=0,$%
\begin{equation}
J(G_{0})=2\int_{0}^{G_{0}(Z)}(1-G_{0}(Z))\left( \frac{Z-G_{0}^{-1}(s)}{Z}%
\right) ds,
\end{equation}%
\begin{equation*}
g_{0}(y)=2\left( 1-G_{0}(y)\right) \left( \frac{Z-y}{Z}\right) 1_{(y\leq Z)},
\end{equation*}%
and%
\begin{equation*}
\nu _{0}(y)=-2\left( \frac{Z-y}{Z}\right) 1_{(y\leq Z)}.
\end{equation*}

\subsection{Kakwani case}

We also have the same conclusion for the Kakwami measure of parameter $k\geq 1$
with $c(x,y)=(x-y)^{k}$ and $\pi (x,y)=y^{k}/x,$

\begin{equation*}
J(G_{0})=(k+1)\int_{0}^{G_{0}(Z)}\left( 1-\frac{s}{G_{0}(Z)}\right)
^{k}\left( \frac{Z-G_{0}^{-1}(s)}{Z}\right) ds,
\end{equation*}%
\begin{eqnarray*}
K(G_{0}) &=&\frac{k(k+1)}{G_{0}(Z)}\int_{0}^{G_{0}(Z)}\left( 1-\frac{s}{%
G_{0}(Z)}\right) ^{k-1}\left( \frac{Z-G_{0}^{-1}(s)}{Z}\right) ds \\
&&+\frac{J(G_{0})}{G_{0}(Z)},
\end{eqnarray*}%
\begin{eqnarray*}
g_{0}(y) &=&\left\{ (k+1)\left[ \left( 1-\frac{G_{0}(y)}{G_{0}(Z)}\right)
^{k}\left( \frac{Z-y}{Z}\right) \right. \right. \\
&&-\left. \left. \frac{J(G_{0})}{G_{0}(Z)}\left( \frac{G_{0}(y)}{G_{0}(Z)}%
\right) ^{k}\right] +K(G_{0})\right\} 1_{(y\leq Z)},
\end{eqnarray*}%
and%
\begin{eqnarray*}
\nu _{0}(y) &=&-\frac{k(k+1)}{G_{0}(Z)}\left[ \left( 1-\frac{G_{0}(y)}{%
G_{0}(Z)}\right) ^{k-1}\left( \frac{Z-y}{Z}\right) \right. \\
&&+\left. \frac{J(G_{0})}{G_{0}(Z)}\left( \frac{G_{0}(y)}{G_{0}(Z)}\right)
^{k-1}\right] 1_{(y\leq Z)}.
\end{eqnarray*}

\subsection{Data-driven applications}

In this note, let us focus on the Sen case, which is more tricky than the
Shorrocks one. We consider the Senegalese database ESAM 1 of 1996 which
includes 3278 households. We first consider the geographical decomposition
into the areas, Dakar is the Capital. We have the Sen measure values for the whole 
 Senegal and for its ten sub-areas.

\begin{center}
\begin{tabular}{|l|c|c|c|c|c|c|}
\hline
Area & Senegal & Kolda & Dakar & Diourbel & Saint-Louis & Louga \\ \hline
Sen Index & 34.71\% & 51.66\% & 22.73\% & 40.16\% & 37.51\% & 34.53\% \\ 
\hline
Size & 3278 & 198 & 1122 & 231 & 314 & 174 \\ \hline
\end{tabular}

\begin{tabular}{|l|c|c|c|c|c|}
\hline
Area & Tambacounda & Kaolack & Thies & Fatick & Ziguinchor \\ \hline
Sen Index & 47.47\% & 37.91\% & 41.31\% & 42.22\% & 39.13\% \\ \hline
Size & 126 & 316 & 401 & 180 & 216 \\ \hline
\end{tabular}
\end{center}

\bigskip

\noindent Let us compute the different variances $\vartheta _{1}^{2}$ $\vartheta _{2}^{2}$ and $%
\vartheta _{3}^{2}$ of Theorem \ref{theo1} with the empirical estimations $p_{i}\approx n_{i}/n,$. We obtain for the geographical decomposability in Senegal : $\vartheta _{1}^{2}+\vartheta _{2}^{2}=0.093195,$ $\vartheta
_{1}^{2}+\vartheta _{3}^{2}=0.093224$ and $gd_{n}=1.25450$ $10^{-3}$ . This
gives the $95\%$-confidence :%
\begin{equation*}
dg\in \lbrack -0.00919\%,0.00117\%],
\end{equation*}

\noindent that is

\begin{equation*}
J(G)\in \lbrack 34.7\%, 34.71\% \rbrack,
\end{equation*}

\noindent We remark the very accurate estimation of the Sen index for the
whole country of Senegal which makes us tell that this index is practically
decomposable in this empirical case. We have already explained that
decomposability does not matter when the distribution is uniform in the
population. It happens that earlier works show that the senegalese date are
well fitted by the lognormal or the Singh-Maddala model for each area with
very similar parameters. Now for a decomposition with respect to the
household chief gender, we get the sen measure values.\\

\begin{center}
\begin{tabular}{|l|c|c|c|}
\hline
Gender & Senegal & Male & female \\ \hline
Sen Index & 34.7 \% & 35.27 \% & 32.62 \% \\ \hline
size & 3278 & 2559 & 919 \\ \hline
\end{tabular}
\end{center}

\bigskip 

\noindent We get here $\vartheta _{1}^{2}+\vartheta _{2}^{2}=1.87$, $%
\vartheta _{1}^{2}+\vartheta _{3}^{2}=1.78$, $gd_{n}=1.496\times 10^{-4}$%
and this $95\%$-confidence :%
\begin{equation*}
dg\in \lbrack -0.00437\%,0.0016\%],
\end{equation*}

\noindent that is

\begin{equation*}
J(G)\in \lbrack 34.696\%,34.704\%],
\end{equation*}

\noindent We get the same conclusion that the gap of decomposability is significantly very low.\\

\noindent We have for the Mauritanian data (EPCV 2004) the following geographical and gender decomposability estimates. For the whole country and its thirteen sub-areas, we have :\\

\begin{table}[htbp]
\centering
\begin{tabular}{|c|c|c|c|c|}
\hline
Area & Mauritanie & Hodh Charghy & Hodh Gharby & Guidimagha \\ \hline
Sen Index & 7,5\% & 6,73\% & 7,59\% & 10,89\% \\ \hline
Size & 9360 & 1211 & 469 & 234 \\ \hline
\end{tabular}

\begin{tabular}{|l|c|c|c|c|c|}
\hline
Area & Adrar & Nouadhibou & Tagant & Tiris Zemmour & Assaba \\ \hline
Sen Index & 5,5\% & 0,83\% & 13,34\% & 2,78\% & 6,49\% \\ \hline
Size & 568 & 585 & 490 & 284 & 514 \\ \hline
\end{tabular}

\begin{tabular}{|c|c|c|c|c|c|}
\hline
Area & Brakna & Trarza & Inchiri & Gorgol & Nouakchott \\ \hline
Sen Index & 11,57\% & 9,12\% & 4,89\% & 12,43\% & 3,49\% \\ \hline
Size & 1190 & 1217 & 205 & 796 & 1597 \\ \hline
\end{tabular}
\vspace{0.2cm}
\end{table}

\bigskip

\noindent $\vartheta _{1}^{2}+\vartheta _{2}^{2}=7,85\times 10^{-2}$, $%
\vartheta _{1}^{2}+\vartheta _{3}^{2}=7,85\times 10^{-2}$ and $gd_{n}=6,40\times
10^{-4}$. This gives the $95\%$-confidence :

\begin{equation*}
dg\in \lbrack -0.00503\%,0.00631\%]
\end{equation*}

\noindent For a stratification with respect to the gender of the chief household, we have :\\
\begin{center}
\begin{tabular}{|l|c|c|c|}
\hline
Gender & Mauritania & Male & female \\ \hline
Sen Index & 7,5 \% & 7,46 \% & 7,64 \% \\ \hline
size & 9360 & 7513 & 1847 \\ \hline
\end{tabular}
\end{center}

\bigskip

\noindent $\vartheta _{1}^{2}+\vartheta _{2}^{2}=5,58\times 10^{-2}$, $\vartheta
_{1}^{2}+\vartheta _{3}^{2}=5,58\times 10^{-2}$, $gd_{n}=3,99\times
10^{-5}$ and the $95\%$-confidence :

\begin{equation*}
dg\in \lbrack -0.004,74\%,0.00482\%],
\end{equation*}

\noindent Our general conclusion is that for all these cases, the sen measure is almost decomposable. But, this does not really matter. The important result is that we are able to have an accurate estimation of the gap of decomposability.
\section{Proofs}

\label{sec6}

\noindent To begin, we need more notations to describe the representation
result of \cite{slrep}, in an appropriate way to our proof. Let $G_{0}\in
\{G,G_{1},...,G_{K}\}$ and let a sample of incomes $\{V_{1},...,V_{m}\}$
from $G_{0}.$ Let $\alpha _{G_{0},m}$ the uniform empirical functional
process based on 
\begin{equation*}
\{G_{0}(V_{1}),...,G_{0}(V_{m})\},
\end{equation*}%
defined by%
\begin{equation*}
\alpha _{G_{0},m}(g_{0})=\frac{1}{\sqrt{m}}\sum_{j=1}^{m}g_{0}(G_{0}(V_{j})-%
\mathbb{E}g_{0}(G_{0}(V_{j})),
\end{equation*}%
and define an other empirical process, called here residual empirical
process,%
\begin{equation}
\beta _{G_{0},m}(\nu _{0})=\frac{1}{\sqrt{m}}\sum_{j=1}^{m}\left\{
G_{G_{0},m}(V_{j})-G_{0}(V_{j})\right\} \nu _{0}(V_{j}),  \label{rep}
\end{equation}%
where $G_{G_{0},m}$ is the empirical distribution function associated with $%
\{V_{1},...,V_{m}\}.$ The representation Theorem of \ Sall and Lo \cite{slrep} establishes under the hypotheses (HD0)-(HD6), for $%
J(G_{0})=H_{c}(G_{0})/H_{\pi }(G_{0})$, 
\begin{equation*}
\sqrt{m}(J_{m}(G_{0})-J(G_{0}))=\alpha _{G_{0},m}(g_{0})+\beta
_{G_{0},m}(\nu _{0})+o_{P}(1)
\end{equation*}

\noindent as $m\rightarrow \infty$, where $g_{0}$ and $\nu _{0}$ are
described in (\ref{g0}) and (\ref{nu0}).

\bigskip

\noindent Before going any further, we should precise the notations for the
global population and the subgroups. For $G=G_{0},$ we drop the subscript $%
G_{0}$ so that $\alpha _{n},$ $\beta _{n},$ $G_{n},$ $J_{n}$ are
respectively the empirical, the residual empirical process (\ref{rep}), the
empirical distribution function and the GPI based on the sample $%
Y_{1},...,Y_{n},$ and $J=J(G)=$ $H_{c}(G)/H_{\pi }(G).$ As well the
functions $g_{0}$ and $\nu _{0}$ are denoted as $g$ and $\nu $ for $G=G_{0}.$
For $G=G_{i},$ $1\leq i\leq K,$ we use the subscript $i$ so that $\alpha
_{i,n_{i}^{\ast }},$ $\beta _{i,n_{i}^{\ast }},$ $G_{i,n_{i}^{\ast }},$ $%
J_{i,n_{i}^{\ast }}$ will respectively denote the empirical, the residual
empirical process (\ref{rep}), the empirical distribution function and the
GPI based on the sample $Y_{i,1},...,Y_{i,n_{i}^{\ast }},$ and $%
J_{i}(G_{i})= $ $H_{c}(G_{i})/H_{\pi }(G_{i}),$ accordingly to the notations
of Section \ref{sec4}, and the functions $g_{0}$ and $\nu _{0}$ are denoted
as $g_{i}$ and $\nu _{i}$ in this case. But sometimes we may feel the
notations so heavy and then lessen them. For example, we only put $%
J_{i}(G_{i})=J(G_{i})$ and $J_{i,n_{i}^{\ast }}(G_{i})=J_{n_{i}^{\ast
}}(G_{i})$, $i \in \{1,..,K \}$.\\

\noindent To begin the proof, we remark that $n^{\ast }(\omega
_{1})=(n_{1}^{\ast }(\omega _{1}),...,n_{K}^{\ast }(\omega _{1}))\rightarrow
_{\mathbb{P}_{1}}\{+\infty \}^{K}$ as $n=n_{1}^{\ast }(\omega
_{1})+...+n_{K}^{\ast }(\omega _{1})\rightarrow \infty .$ We then get 
\begin{equation}
\sqrt{n}(J_{n}(G)-J(G))=\alpha _{n}(g)+\beta _{n}(\nu )+o_{P}(1):=\gamma
_{n}+o_{P}(1)  \label{l1}
\end{equation}%
and for any $1\leq i\leq K$, 
\begin{equation}
\sqrt{n_{i}^{\ast }}(J_{n_{i}^{\ast }}(G_{i})-J(G_{i}))=\alpha
_{i,n_{i}^{\ast }}(g_{i})+\beta _{i,n_{i}^{\ast }}(\nu
_{i})+o_{P}(1):=\gamma _{i,n_{i}^{\ast }}+o_{P}(1)  \label{l2}
\end{equation}

\bigskip 

\noindent Now we use the intermediate centering coefficient 
\begin{equation*}
gd_{0,n}=J(G)-\sum_{i=1}^{K}\frac{n_{i}^{\ast }}{n}J(G_{i}).
\end{equation*}%
to get from (\ref{l1}) and (\ref{l2}) 
\begin{multline}
\left\vert \sqrt{n}(gd_{n}-gd_{0,n})(\omega _{1},\omega _{2})-\left\{ \gamma
_{n}-\sum_{j=1}^{K}\left( \frac{n_{i}^{\ast }}{n}\right) ^{1/2}\gamma
_{i,n_{i}}\right\} (\omega _{1},\omega _{2})\right\vert  \label{approx_n1} \\
\rightarrow _{P_{1}\otimes P_{2}}0,
\end{multline}%
as $n\rightarrow \infty $. Then, we have
\begin{eqnarray*}
S_{n}^{\ast } &=&\gamma _{n}(g,\nu )-\sum_{j=1}^{K}\left( \frac{n_{i}^{\ast }%
}{n}\right) ^{1/2}\gamma _{i,n_{i}^{\ast }}(g_{i},\nu _{i}) \\
&=&\alpha _{n}(g)-\sum_{j=1}^{K}\left( \frac{n_{i}^{\ast }}{n}\right)
^{1/2}\alpha _{i,n_{i}^{\ast }}(g_{i})+\beta _{n}(\nu )-\sum_{j=1}^{K}\left( 
\frac{n_{i}^{\ast }}{n}\right) ^{1/2}\beta _{i,n_{i}^{\ast }}(\nu _{i}).
\end{eqnarray*}%
Remark that%
\begin{equation*}
\alpha _{n}(g)=\frac{1}{\sqrt{n}}\sum_{j=1}^{n}\left( g(Y_{j})-\mathbb{E}%
g(Y)\right) =\sqrt{n}\left( \frac{1}{n}\sum_{j=1}^{n}g(Y_{j})-\mathbb{E}%
g(Y)\right)
\end{equation*}%
\begin{equation*}
=:\sqrt{n}\left( \frac{1}{n}\sum_{j=1}^{n}g(Y_{j})-\sum_{i=1}^{K}\frac{%
n_{i}^{\ast }}{n}\mathbb{E}g(Y^{i})\right) +D^{\ast }(n,1)
\end{equation*}%
with

$$
D(n,1)=\sum_{i=1}^{K}\frac{n_{i}-np_{i}}{\sqrt{np_{i}}}\sqrt{p_{i}}\mathbb{E}g(Y^{i}),
$$

\noindent and\\

$$
D^{\ast }(n,1)=\sum_{i=1}^{K}\frac{n_{i}^{\ast }-np_{i}}{\sqrt{np_{i}}}%
\mathbb{E}g(Y^{i})\sqrt{p_{i}}.
$$

\noindent This leads to

\begin{eqnarray*}
S_{n}^{\ast } &=&\sqrt{n}\left( \frac{1}{n}\sum_{j=1}^{n}g(Y_{j})-%
\sum_{i=1}^{K}\frac{n_{i}^{\ast }}{n}\mathbb{E}g(Y^{i})\right)
-\sum_{j=1}^{K}\left( \frac{n_{i}^{\ast }}{n}\right) ^{1/2}\alpha
_{i,n_{i}^{\ast }}(g_{i}) \\
&&+\beta _{n}(\nu )-\sum_{j=1}^{K}\left( \frac{n_{i}^{\ast }}{n}\right)
^{1/2}\beta _{i,n_{i}^{\ast }}(\nu _{i})+D^{\ast }(n,1).
\end{eqnarray*}%

\noindent Now, by denoting 
\begin{equation*}
C^{\ast }(n,1)=\sqrt{n}\left( \frac{1}{n}\sum_{j=1}^{n}g(Y_{j})-%
\sum_{i=1}^{K}\frac{n_{i}^{\ast }}{n}\mathbb{E}g(Y^{i})\right)
-\sum_{i=1}^{K}\left( \frac{n_{i}^{\ast }}{n}\right) ^{1/2}\alpha
_{i,n_{i}^{\ast }}\left( g_{i}\right),
\end{equation*}%

\noindent one has
\begin{equation}
C^{\ast }(n,1)=\sum_{i=1}^{K}\left( \frac{n_{i}^{\ast }}{n}\right) ^{1/2}%
\frac{1}{\sqrt{n_{i}^{\ast }}}\sum_{j=1}^{n_{i}^{\ast}}\left[ \left( g-g_{i}\right)
\left( Y_{i,j}\right) -\mathbb{E}\left( g-g_{i}\right) (Y^{i})\right] .
\label{c1}
\end{equation}%
we get%
\begin{equation}
S_{n}^{\ast }=C^{\ast }(n,1)+D^{\ast }(n,1)+\beta _{n}(\nu
)-\sum_{j=1}^{K}\left( \frac{n_{i}^{\ast }}{n}\right) ^{1/2}\beta
_{i,n_{i}^{\ast }}(\nu _{i}).  \label{for1a}
\end{equation}

\noindent Further one has
\begin{equation}
\sum_{j=1}^{K}\left( \frac{n_{i}^{\ast }}{n}\right) \beta _{i,n_{i}}^{\ast
}(\nu _{i})=\frac{1}{\sqrt{n}}\sum_{i=1}^{K}\sum_{j=1}^{n_{i}^{\ast
}}[G_{i,n_{i}^{\ast }}(Y_{ij})-G_{i}(Y_{ij}))]\nu _{i}(Y_{ij})  \label{beta1}
\end{equation}%
But%
\begin{equation*}
G(Y_{ij})=\sum_{h=1}^{K}p_{h}G_{h}(Y_{ij}),
\end{equation*}%
and for $x\in \mathbb{R},$%
\begin{equation*}
G_{n}(x)=\frac{1}{n}\sum_{i=1}^{n}1_{(Y_{j}\leq x)}=\frac{1}{n}%
\sum_{i=1}^{K}\sum_{j=1}^{n_{i}^{\ast }}1_{(Y_{i,j}\leq x)}
\end{equation*}%
\begin{equation*}
=\sum_{i=1}^{K}(\frac{n_{i}^{\ast }}{n})\frac{1}{n_{i}^{\ast }}%
\sum_{j=1}^{n_{i}^{\ast }}1_{(Y_{i,j}\leq x)}=\sum_{i=1}^{K}\frac{%
n_{i}^{\ast }}{n}G_{i,n_{i}^{\ast }}(x).
\end{equation*}%
Thus 
\begin{equation*}
\beta _{n}(\nu )=\frac{1}{\sqrt{n}}\sum_{i=1}^{K}\sum_{j=1}^{n_{i}^{\ast }}%
\left[ \sum_{h=1}^{K}(\frac{n_{h}^{\ast }}{n})G_{i,n_{h}^{\ast
}}(Y_{ij})-p_{h}G_{h}(Y_{ij})\right] \nu (Y_{ij}).
\end{equation*}%
From this, we put and subtract $\sum_{h=1}^{k}(\frac{n_{h}^{\ast }}{n}%
)G_{h}(Y_{ij})$ to have 
\begin{eqnarray*}
\beta _{n}(\nu ) &=&\frac{1}{\sqrt{n}}\sum_{i=1}^{K}\sum_{j=1}^{n_{i}^{\ast
}}\left[ \sum_{h=1}^{K}\left( \frac{n_{h}^{\ast }}{n}\right)
G_{i,n_{h}^{\ast }}(Y_{ij})-\sum_{h=1}^{K}\left( \frac{n_{h}^{\ast }}{n}%
\right) G_{h}(Y_{ij})\right] \nu (Y_{ij}) \\
&&+\frac{1}{\sqrt{n}}\sum_{i=1}^{K}\sum_{j=1}^{n_{i}^{\ast }}\left[
\sum_{h=1}^{K}\left( \frac{n_{h}^{\ast }}{n}-p_{h}\right) G_{h}(Y_{ij})%
\right] \nu (Y_{ij})
\end{eqnarray*}

\begin{equation}
=\frac{1}{\sqrt{n}}\sum_{i=1}^{K}\sum_{j=1}^{n_{i}}\sum_{h=1}^{K}\left( 
\frac{n_{h}^{\ast }}{n}\right) \left[ G_{n_{h}}(Y_{ij})-G_{h}(Y_{ij})\right]
\nu (Y_{ij})  \label{beta2}
\end{equation}

\begin{equation*}
+\frac{1}{\sqrt{n}}\sum_{i=1}^{K}\sum_{j=1}^{n_{i}}\left[ \sum_{h=1}^{K}%
\left( \frac{n_{h}^{\ast }}{n}-p_{h}\right) G_{h}(Y_{ij})\right] \nu
(Y_{ij}).
\end{equation*}%
Now we put together (\ref{beta1}) and (\ref{beta2}), while separating the
two cases $h=i$ and $h\neq i$ in (\ref{beta2}) to get%
\begin{gather*}
\beta _{n}(\nu )-\sum_{j=1}^{K}\left( \frac{n_{i}^{\ast }}{n}\right)
^{1/2}\beta _{i,n_{i}}(\nu _{i})= \\
\sum_{i=1}^{K}\left( \frac{n_{i}^{\ast }}{n}\right) ^{1/2}\left\{ \frac{1}{%
\sqrt{n_{i}^{\ast }}}\sum_{j=1}^{n_{i}}\left\{ G_{i,n_{i}^{\ast
}}(Y_{ij})-G_{i}(Y_{ij})\right\} \left( \frac{n_{i}^{\ast }}{n}\nu -\nu
_{i}\right) (Y_{ij})\right\} \\
+\sum_{i=1}^{K}\left( \frac{n_{i}^{\ast }}{n}\right) ^{1/2}\sum_{h\neq i}^{K}%
\frac{n_{h}^{\ast }}{n}\frac{1}{\sqrt{n_{i}^{\ast }}}\sum_{j=1}^{n_{i}^{\ast
}}\left[ G_{n_{h}^{\ast }}(Y_{ij})-G_{h}(Y_{ij})\right] \nu (Y_{ij}) \\
+\frac{1}{\sqrt{n}}\sum_{i=1}^{K}\sum_{j=1}^{n_{i}^{\ast }}\left[
\sum_{h=1}^{K}\left( \frac{n_{h}^{\ast }}{n}-p_{h}\right) G_{h}(Y_{ij})%
\right] \nu (Y_{ij})
\end{gather*}

\begin{equation}
=:C^{\ast }(n,2)+C^{\ast }(n,3)+D^{\ast }(n,2),  \label{for1b}
\end{equation}%
with

\begin{equation}
C^{\ast }(n,2)=\sum_{i=1}^{K}\left( \frac{n_{i}^{\ast }}{n}\right)
^{1/2}\left\{ \frac{1}{\sqrt{n_{i}^{\ast }}}\sum_{j=1}^{n_{i}}\left\{
G_{i,n_{i}^{\ast }}(Y_{ij})-G_{i}(Y_{ij})\right\} \left( \frac{n_{i}^{\ast }%
}{n}\nu -\nu _{i}\right) (Y_{ij})\right\} ,  \label{c2}
\end{equation}%
\noindent and 
\begin{equation}
C^{\ast }(n,3)=\sum_{i=1}^{K}\left( \frac{n_{i}^{\ast }}{n}\right)
^{1/2}\sum_{h\neq i}^{K}\frac{n_{h}^{\ast }}{n}\frac{1}{\sqrt{n_{i}^{\ast }}}%
\sum_{j=1}^{n_{i}^{\ast }}\left[ G_{n_{h}^{\ast }}(Y_{ij})-G_{h}(Y_{ij})%
\right] \nu (Y_{ij}).  \label{c3}
\end{equation}

\noindent We arrive, by comparing (\ref{for1a}) and (\ref{for1b}), at%
\begin{equation}
S_{n}^{\ast }=C^{\ast }(n,1)+C^{\ast }(n,2)+C^{\ast }(n,3)+D^{\ast
}(n,1)+D^{\ast \ast}(n,2).  \label{sn}
\end{equation}%
Let us have a look at 
\begin{equation*}
D^{\ast \ast }(n,2)=\sqrt{n}\sum_{h=1}^{K}\left( \frac{n_{h}^{\ast }}{n}%
-p_{h}\right) \left\{ \sum_{i=1}^{K}\left( \frac{n_{i}^{\ast }}{n}\right) 
\frac{1}{n_{i}^{\ast }}\sum_{j=1}^{n_{i}^{\ast }}G_{h}(Y_{ij})\nu
(Y_{ij})\right\} .
\end{equation*}%
By the weak law of large numbers%
\begin{equation*}
\left\{ \sum_{i=1}^{K}\left( \frac{n_{i}^{\ast }}{n}\right) \frac{1}{%
n_{i}^{\ast }}\sum_{j=1}^{n_{i}^{\ast }}G_{h}(Y_{ij})\nu (Y_{ij})\right\}
\rightarrow _{\mathbb{P}}\sum_{i=1}^{K}p_{i}\mathbb{E}G_{h}(Y^{i})\nu
(Y^{i})=H_{h}.
\end{equation*}%
That is%
\begin{equation*}
D^{\ast \ast}(n,2)=\sum_{h=1}^{K}\left( \frac{n_{h}^{\ast }-np_{h}}{\sqrt{np_{h}}%
}\right) H_{h}\sqrt{p_{h}}+o_{P}(1).
\end{equation*}

$$
=: D^{\ast}(n,2) + o_{P}(1).
$$

\noindent Finally

\begin{equation}
gd_{n}^{\ast }=S_{n}^{\ast }+\sqrt{n}(gd_{0,n}-gd)  \label{sntosnt}.
\end{equation}

\noindent Hence

\begin{equation*}
gd_{n}^{\ast }=C^{\ast }(n,1)+C^{\ast }(n,2)+C^{\ast }(n,3)
\end{equation*}

\begin{equation*}
+D^{\ast }(n,1)+D^{\ast }(n,2)-\sum_{i=1}^{K}\left( \frac{n_{i}^{\ast
}-np_{i}}{\sqrt{np_{i}}}\right) J_{i}(G_{i})\sqrt{p_{i}}+o_{P}(1),
\end{equation*}%
\begin{equation}
=:C^{\ast }(n)+D^{\ast }(n)+o_{P}(1). \label{reste}
\end{equation}%
with%
\begin{equation}
C^{\ast }(n)=C^{\ast }(n,1)+C^{\ast }(n,2)+C^{\ast }(n,3)  \label{c}
\end{equation}%
and%
\begin{equation*}
D^{\ast }(n)=D^{\ast }(n,1)+D^{\ast }(n,2)-\sum_{i=1}^{K}\left( \frac{%
n_{i}^{\ast }-np_{i}}{\sqrt{np_{i}}}\right) J_{i}(G_{i})\sqrt{p_{i}}
\end{equation*}%
\begin{equation*}
=\sum_{i=1}^{K}\left( \frac{n_{i}^{\ast }-np_{i}}{\sqrt{np_{i}}}%
\right) (H_{i}+\mathbb{E}g(Y^{i})-J_{i}(G_{i}))\sqrt{p_{i}}
\end{equation*}

\begin{equation*}
=:\sum_{i=1}^{K}\left( \frac{n_{i}^{\ast }-np_{i}}{\sqrt{np_{i}}}%
\right) F_{i}\sqrt{p_{i}}.
\end{equation*}

\bigskip

\noindent \textit{We have now to prove that }$gd_{n}^{\ast }=\sqrt{n}%
(gd_{n}-gd)$\textit{\ weakly converges to a }$\mathcal{N}(0,\vartheta
_{1}^{2}+\vartheta _{2}^{2})$ random variable. For this it suffices, based on  \ref{reste}, to prove that
$S_{n}^{\ast \ast}=C^{\ast }(n)+D^{\ast }(n)$ converges to $\mathcal{N}(0,\vartheta
_{1}^{2}+\vartheta _{2}^{2})$. Now put
\begin{equation*}
\mathbb{N}(K)=\{\overline{n}=(n_{1},...n_{K}),n_{i}\geq
0,n_{1}+...,n_{K}=n\}.
\end{equation*}%
Since $n^{\ast }=(n_{1}^{\ast },...n_{K}^{\ast })\rightarrow
_{P_{1}}\{\infty \}^{K},$ we find for a fixed $\varepsilon >0$, $K$ positive
numbers $N_{i}$ $(1\leq i\leq K)$ such that for $n_{i}\geq N_{i}$ $(1\leq
i\leq K),$ which implies that $n\geq N=N_{1}+...+N_{K},$%
\begin{equation*}
\mathbb{P}(\exists (1\leq i\leq K),n_{i}^{\ast }<N_{i})<\varepsilon .
\end{equation*}%
Let%
\begin{equation*}
\mathcal{N}(K,1)=\mathbb{N}(K)\cap \{\overline{n}=(n_{1},...n_{K}),\exists
(1\leq i\leq K),n_{i}<N_{i}\}
\end{equation*}%
and $\mathbb{N}(K,2)=\mathbb{N}(K)\diagdown \mathbb{N}(K,1).$ We remark that
conditionally on $(n^{\ast }=\overline{n})$, $C^{\ast }(n)$ becomes $C(n),$
does not depend on $\omega _{1}$ and only include the independent random
variables $\{Y_{i,j},1\leq j\leq n_{i},1\leq i\leq K\}$. From Lemma \ref%
{lem2} below, we have 
\begin{equation*}
C(n)\rightarrow \mathcal{N}(0,\vartheta _{1}^{2}).
\end{equation*}%
Also conditionally on $(n^{\ast }=\overline{n})$, $D^{\ast }(n)$ becomes
$D^(n)$ and we denote it $D(n)$. Now for $h^{2}=-1,$%
$$
\psi _{S_{n}^{\ast \ast}}(t) =\mathbb{E}(\exp (htS_{n}^{\ast \ast}))
$$

$$
=\sum_{\overline{n}\in \mathcal{N}(K)}P(n^{\ast }=\overline{n})\mathbb{E}%
(\exp (htC^{\ast }(n)+htD^{\ast }(n)) \diagup (n^{\ast }=\overline{n}))
$$

\begin{equation*}
=\sum_{\overline{n}\in \mathcal{N}(K)}P(n^{\ast }=\overline{n})\mathbb{E}%
(\exp (htD(n))\text{ }\mathbb{E}(\exp (htC^{\ast }(n))\diagup (n^{\ast }=%
\overline{n})).
\end{equation*}%
Recall that, by the classical limiting law of the multinomial $K$-vector, 
\begin{equation*}
D^{\star }(n)\rightarrow D=\sum_{i=1}^{K}Z_{i}F_{i}\sqrt{p_{i}},
\end{equation*}%
where $(Z_{1},...,Z_{K})^{t}$ is a Gaussian vector with $Var(Z_{i})=1-p_{i}$
and $Cov(Z_{i},Z_{j})=-\sqrt{p_{i}p_{j}},$ for $i\neq j.$ Then 
\begin{equation*}
D^{\ast}(n) \rightarrow \mathcal{N}(0,\vartheta _{2}^{2}),
\end{equation*}%
with 
\begin{equation*}
\vartheta _{2}^{2}=\sum_{h=1}^{K}F_{h}^{2}p_{h}(1-p_{h})-\sum_{1\leq h\neq
k\leq K}F_{h}F_{k}p_{h}p_{k}
\end{equation*}%
\begin{equation*}
=\sum_{h=1}^{K}F_{h}{}^{2}p_{h}-\left( \sum_{h=1}^{K}F_{h}p_{h}\right) ^{2}.
\end{equation*}%
We remark that this is the variance of the function $F_{h}$ of $h\in \lbrack
1,K]$ with respect to the probability measure $\sum_{1\leq h\leq
K}p_{h}\delta _{h}$.

\noindent Put now 
\begin{equation*}
\mathbb{N}(K,1)=\mathbb{N}(K)\cap \{\overline{n}=(n_{1},...n_{K}),\exists
(1\leq i\leq K),n_{i}<N_{i}\}
\end{equation*}%
and $\mathbb{N}(K,2)=\mathbb{N}(K)\diagdown \mathbb{N}(K,1)$. Then 
\begin{equation*}
\sum_{\overline{n}\in \mathbb{N}(K)}\exp (htD(n))\mathbb{P}(n^{\ast }=%
\overline{n})\mathbb{E}(\exp (htC(n))))=B(n,1)+B(n,2)
\end{equation*}%
with 
\begin{equation*}
\left\vert B(n,1)\right\vert =\left\vert \sum_{\overline{n}\in \mathbb{N}%
(K,1)}\exp (htD(n))\mathbb{P}(n^{\ast }=\overline{n})\mathbb{E}(\exp
(htC(n)))\right\vert
\end{equation*}%
\begin{equation}
\leq \mathbb{P}(\exists (1\leq i\leq K),n_{i}^{\ast }<N_{i})\rightarrow 0,
\label{approx2}
\end{equation}%
and 
\begin{equation}
\left\vert B(n,2)-\sum_{\overline{n}\in \mathbb{N}(K,2)}\exp (-(\vartheta
_{1}t)^{2}/2)\exp (htD(n))P(n^{\ast }=\overline{n})\right\vert
\label{approx3}
\end{equation}%
\begin{equation*}
\leq \varepsilon \sum_{\overline{n}\in \mathbb{N}(K,2)}P(n^{\ast }=\overline{%
n})\leq \varepsilon .
\end{equation*}%
Finally, for 
\begin{equation}
B^{\ast }(n,2)=\sum_{\overline{n}\in \mathbb{N}(K,2)}\exp (-(\vartheta
_{1}t)^{2}/2)\exp (htD(n))P(n^{\ast }=\overline{n}),  \label{approx4}
\end{equation}%
we are able to use (\ref{approx4}) and to get 
\begin{multline}
\lim \sup_{n\rightarrow \infty }\left\vert B^{\ast }(n,2)-\sum_{\overline{n}%
\in \mathbb{N}(K)}\exp (htD(n))P(n^{\ast }=\overline{n})\mathbb{E}(\exp
(-(\vartheta _{1}t)^{2}/2))\right\vert \\
=0.  \notag
\end{multline}%
But 
\begin{equation}
\mathbb{E}\exp (thD^{\ast }(n))=\sum_{\overline{n}\in \mathbb{N}(K)}\exp
(htD^{\ast }(n)/(n^{\ast }=\overline{n}))P(n^{\ast }=\overline{n})
\end{equation}%
\begin{equation*}
=\sum_{\overline{n}\in \mathbb{N}(K)}\exp (htD(n))P(n^{\ast }=\overline{n}%
)\rightarrow \exp (-(\vartheta _{2}t)^{2}/2))
\end{equation*}%
By putting together the previous formulas, and by letting $\varepsilon
\downarrow 0,$ we arrive at 
\begin{equation*}
\psi _{d_{n}^{\ast \ast}}(t)\rightarrow \exp (-(\vartheta _{1}^{2}+\vartheta
_{2}^{2})t^{2}/2).
\end{equation*}

\noindent This proves the asymptotic normality of $dg_{n}^{\ast }$ of the
theorem corresponding to $S_{n}^{\ast \ast}$. That of $dg_{n,0}^{\ast }$
corresponds to $S_{n}^{\ast }$. This latter is achieved by omitting the term 
$\sqrt{n}\sum_{i=1}^{K}(\frac{n_{i}^{\ast }}{n}-p_{i})J_{i}(G_{i})$ in (\ref%
{sntosnt}). This leads to $M_{h}$ obtained from $F_{h}$ by dropping $%
J_{i}(G_{i})$. This completes the proofs.\\

\noindent We now prove this lemma used in the proof.

\begin{lemma}
\label{lem2} Let $C(n)=C(n,1)+C(n,2)+C(n,3)$, where the $C(n,i)$ are
respectively defined in (\ref{c1}), (\ref{c2}) and (\ref{c3}) for $i=1,2,3$.
Then, as $n\rightarrow +\infty ,$%
\begin{equation*}
C(n)\leadsto \mathcal{N}(0,\vartheta _{1}^{2}).
\end{equation*}
\end{lemma}

\begin{proof}
\bigskip Recall that%
\begin{equation}
C(n)=C(n,1)+C(n,2)+C(n,3).  \label{ca}
\end{equation}

\noindent Let for each $i\in \lbrack 1,K],$ $\mathbb{G}_{n_{i}}(i,f)$ be the
functional empirical process based on $\{G_{i}(Y_{i,j}),1\leq i\leq
n_{i}\},1\leq i\leq K\}.$ We consider the three terms in (\ref{ca}), $\ $%
that is the $C(n,i)$, $1\leq i\leq 3,$ defined in (\ref{c1}), (\ref{c2}) and
in (\ref{c3}), and prove that each of them converges to a random variable $%
C(i)$ depending on the limiting Gaussian processes $\mathbb{G(}i,\cdot )$ of 
$\mathbb{G}_{n_{i}}(i,\cdot )$. This is enough to prove the asymptotic
normality. The variance $\vartheta _{1}^{2}$ will be nothing else but that
of $C(1)+C(2)+C(3)$. Firstly, we treat $C(n,1).$ Remark that conditionally
on $(n^{\ast }=\overline{n}),$ the random sequences $\{Y_{i,j},1\leq i\leq
n_{i},1\leq i\leq K\}$ are independent and only depend on the $\omega
_{2}\in \Omega _{2}.$ We have%
\begin{equation*}
\sum_{i=1}^{K}\left( \frac{n_{i}}{n}\right) ^{1/2}\alpha _{n_{i}}(g_{i})=%
\frac{1}{\sqrt{n}}\left[ \sum_{i=1}^{K}\sum_{j=1}^{n_{i}}g_{i}(Y_{ij})-%
\sum_{i=1}^{K}n_{i}\mathbb{E}(g_{i}(Y^{i}))\right]
\end{equation*}%
\begin{equation*}
=\sqrt{n}\left[ \frac{1}{n}\sum_{i=1}^{K}\sum_{j=1}^{n_{i}}g_{i}(Y_{ij})-%
\sum_{i=1}^{K}\left( \frac{n_{i}}{n}\right) \mathbb{E}\left(
g_{i}(Y^{i})\right) \right] ,
\end{equation*}%
and%
\begin{equation*}
\alpha _{n}(g,1)=\sqrt{n}\left( \frac{1}{n}\sum_{j=1}^{n}g(Y_{j})-%
\sum_{i=1}^{K}\left( \frac{n_{i}}{n}\right) \mathbb{E}\left( g(Y^{i})\right)
\right)
\end{equation*}%
\begin{equation*}
=\sqrt{n}\left( \frac{1}{n}\sum_{i=1}^{K}\sum_{j=1}^{n_{i}}g(Y_{ij})-%
\sum_{i=1}^{K}\left( \frac{n_{i}}{n}\right) \mathbb{E}\left( g(Y^{i})\right)
\right) .
\end{equation*}%
Then, by (\ref{c1}) and replacing $n_{i}^{\ast}$ by $n_{i}$, $i=1,...,K$, we get 

\begin{equation*}
C(n,1)=\alpha _{n}(g,1)-\sum_{i=1}^{K}\left( \frac{n_{i}}{n}\right) \alpha
_{n_{i}}(g_{i})
\end{equation*}%

\begin{equation}
=\sum_{i=1}^{K}\left( \frac{n_{i}}{n}\right) ^{1/2}\left\{ \frac{1}{\sqrt{%
n_{i}}}\sum_{j=1}^{n_{i}}\left\{ \left( g-g_{i}\right) (Y_{ij})-\mathbb{E}%
\left( g-g_{i}\right) (Y^{i}))\right\} \right\} .  \label{c1a}
\end{equation}%

\noindent This implies that%

\begin{equation*}
C(n,1)=\sum_{i=1}^{K}\left( \frac{n_{i}}{n}\right) ^{1/2}\mathbb{G}%
_{n_{i}}\left( i,\left( g-g_{i}\right) G_{i}^{-1}\right) .
\end{equation*}%
We finally have that%
\begin{equation*}
C(n,1)\rightarrow C(1)=\sum_{i=1}^{K}p_{i}^{1/2}\mathbb{G}%
(i,(g-g_{i})G_{i}^{-1}).
\end{equation*}%
Since the $\mathbb{G}\left( i,\left( g-g_{i}\right) G_{i}^{-1}\right) $ are
independent, centered and Gaussian, we get that 
\begin{equation*}
A_{1}=\mathbb{E}C^{2}(1)=\sum_{i=1}^{K}p_{i}\mathbb{EG}%
^{2}(i,(g-g_{i})G_{i}^{-1})
\end{equation*}%
\begin{equation*}
=\sum_{i=1}^{K}p_{i}\left\{ \mathbb{E}(g-g_{i})^{2}(Y^{i})-(\mathbb{E}%
(g-g_{i})(Y^{i}))^{2}\right\} .
\end{equation*}%

\noindent In the sequel we take

\begin{equation*}
g_{0}\left( x\right) =\overline{g}_{0}\left( x\right) \times e(x)\text{ and }%
\nu_{0}\left( x\right) =\overline{\nu }_{0}\left( x\right) \times e(x),
\end{equation*}%

\noindent and

\begin{equation*}
\left( g_{0},\nu _{0}\right) \in \left( g,g_{1},...,g_{K}\right) \times
\left( \nu ,\nu _{1},...,\nu _{K}\right) \text{ and }i=1,...,K.
\end{equation*}

\noindent Then we arrive

\begin{equation*}
A_{1}=\sum_{i=1}^{K}p_{i}\left\{ \int_{0}^{G_{i}(Z)}(\overline{g}-\overline{g%
}_{i})^{2}(G_{i}^{-1}(t))dt-\left( \int_{0}^{G_{i}(Z)}(\overline{g}-%
\overline{g}_{i})(G_{i}^{-1}(t))dt\right) ^{2}\right\} .
\end{equation*}%
Secondly, one has%
\begin{equation*}
C(n,2)=\sum_{i=1}^{K}\left( \frac{n_{i}}{n}\right) ^{1/2}\left\{ \frac{1}{%
\sqrt{n_{i}}}\sum_{j=1}^{n_{i}}\left\{
G_{i,n_{i}}(Y_{ij})-G_{i}(Y_{ij})\right\} \left( \frac{n_{i}}{n}\nu -\nu
_{i}\right) (Y_{ij})\right\} .
\end{equation*}%
We have%
\begin{equation*}
\frac{1}{\sqrt{n_{i}}}\sum_{j=1}^{n_{i}}\left\{
G_{n_{i}}(Y_{ij})-G_{i}(Y_{ij})\right\} \left( \frac{n_{i}}{n}\nu -\nu
_{i}\right) (Y_{ij})
\end{equation*}%
\begin{equation*}
=\int_{0}^{1}-\varepsilon _{n_{i}}(i,s)(p_{i}\nu -\nu
_{i})(G_{i}^{-1}(s))ds+o_{P}(1)
\end{equation*}%
\begin{equation*}
=\int_{0}^{1}\mathbb{G}_{n_{i}}(i,s)(p_{i}\nu -\nu
_{i})(G_{i}^{-1}(s))ds+o_{P}(1)
\end{equation*}%
\begin{equation*}
\rightarrow \int_{0}^{1}\mathbb{G}(i,s)(p_{i}\nu -\nu _{i})(G_{i}^{-1}(s))ds,
\end{equation*}%
and thus%
\begin{equation}
C(n,2)\rightarrow C(2)=\sum_{i=1}^{K}p_{i}^{1/2}\int_{0}^{1}\mathbb{G}%
(i,s)(p_{i}\nu -\nu _{i})(G_{i}^{-1}(s))ds.  \label{l02}
\end{equation}%
Finally, one has
\begin{equation*}
C(n,3)=\sum_{i=1}^{K}\left( \frac{n_{i}}{n}\right) ^{1/2}\sum_{h\neq i}^{K}%
\frac{n_{h}}{n}\frac{1}{\sqrt{n_{i}}}\sum_{j=1}^{n_{i}}\left[
G_{n_{h}}(Y_{ij})-G_{h}(Y_{ij})\right] \nu (Y_{ij}).
\end{equation*}%
But, for each fixed $i\in \{1,..,K\},$%
\begin{equation*}
\left\{ \frac{1}{\sqrt{n_{i}}}\sum_{j=1}^{n_{i}}\left[
G_{n_{h}}(Y_{ij})-G_{h}(Y_{ij})\right] \nu (Y_{ij})\right\}
\end{equation*}%
\begin{equation*}
=\int_{0}^{1}\sqrt{n_{i}}\left\{
G_{n_{h}}(G_{i}^{-1}(V_{n_{i}}(i,s)))-G_{h}(G_{i}^{-1}(V_{n_{i}}(i,s)))%
\right\} \times \nu (G_{i}^{-1}(V_{n_{i}}(i,s)))ds.
\end{equation*}%
We remember that $\nu $ is of the form%
\begin{equation*}
\nu (y)=\nu _{a}(y)1_{(y\leq Z)}
\end{equation*}%
where $\nu _{a}$ is continuous on compact sets $[0,L]$, $L>0$. Since, as $n\rightarrow \infty ,$ 
\begin{equation*}
\sup_{s\in (0,1)}\left\vert V_{n_{i}}(i,s)-s\right\vert \rightarrow 0,\text{ 
}a.s,
\end{equation*}%
we see that, for large values of $n$, theses integrals are performed at most
on some interval $[0,G_{i}(Z)+\varepsilon ],$ which includes those $s$
satisfying $V_{n_{i}}(i,s)\leq G_{i}(Z).$ By the assumptions, the functions $%
\nu _{a}$ and $G$ are continuous on such compact sets. Thus%
\begin{equation*}
\left\{ \frac{1}{\sqrt{n_{i}}}%
\sum_{j=1}^{n_{i}}[G_{n_{h}}(Y_{ij})-G_{h}(Y_{ij})]\nu (Y_{ij})\right\}
\end{equation*}%
\begin{equation*}
=\sqrt{\frac{n_{i}}{n_{h}}}\int_{0}^{1}\mathbb{G}%
_{h,n_{h}}(h,G_{h}(G_{i}^{-1}(V_{n_{i}}(i,s)))\times \nu
(G_{i}^{-1}(V_{n_{i}}(i,s)))ds
\end{equation*}%
\begin{equation*}
=\sqrt{\frac{n_{i}}{n_{h}}}\int_{0}^{1}\mathbb{G}%
_{h,n_{h}}(h,G_{h}(G_{i}^{-1}(V_{n_{i}}(i,s)))\times \nu
(G_{i}^{-1}(s))ds+o_{P}(1).
\end{equation*}%
Next%
\begin{equation*}
=\sqrt{\frac{n_{i}}{n_{h}}}\int_{0}^{1}G_{n_{h}}(h,G_{h}(G_{i}^{-1}(s))%
\times \nu (G_{i}^{-1}(s))ds+R_{n}+o_{P}(1),
\end{equation*}%
with%
\begin{equation*}
R_{n}=\int_{0}^{1}\left\{ \mathbb{G}%
_{h,n_{h}}(h,G_{h}(G_{i}^{-1}(V_{n_{i}}(i,s)))-\mathbb{G}%
_{h,n_{h}}(h,G_{h}(G_{i}^{-1}(s))\right\} \times \nu (G_{i}^{-1}(s))ds
\end{equation*}%
and%
\begin{equation*}
\left\vert R_{n}\right\vert \leq \int_{0}^{G_{i}(Z)+\varepsilon }\left\vert 
\mathbb{G}_{h,n_{h}}(h,G_{h}(G_{i}^{-1}(V_{n_{i}}(i,s)))-\mathbb{G}%
_{h,n_{h}}(h,G_{h}(G_{i}^{-1}(s))\right\vert \times \nu (G_{i}^{-1}(s))ds.
\end{equation*}%
We surely have, by continuity of $G_{h}$ on $\left( 0,G_{i}^{-1}\left(
G(Z)+\varepsilon \right) \right) ,$ 
\begin{equation*}
\sup_{s\leq G_{i}(Z)+\varepsilon }\left\vert
G_{h}(G_{i}^{-1}(V_{n_{i}}(i,s)))-G_{h}(G_{i}^{-1}(s))\right\vert
=a_{n}\rightarrow 0.
\end{equation*}%
We obtain here a continuous modulus of the uniform empirical process (see
Shorrack and wellner \cite{shwell}, page 531) and then%
\begin{equation*}
\sup_{s\leq G_{i}(Z)+\varepsilon }\left\vert \left\{
G_{h,n_{h}}(h,G_{h}(G_{i}^{-1}(V_{n_{i}}(i,s)))-G_{h,n_{h}}(h,G_{h}(G_{i}^{-1}(s))\right\} \right\vert =O(%
\sqrt{-a_{n}\log a_{n}}).
\end{equation*}%
We finally get%
\begin{equation*}
R_{n}=O\left( \sqrt{-a_{n}\log a_{n}}\right) \int_{0}^{1}\nu
(G_{i}^{-1}(s))ds\rightarrow 0
\end{equation*}%
and we arrive at%
\begin{equation}
C(n,3)\rightarrow C(3)=\sum_{i=1}^{K}\sqrt{p_{i}}\sum_{h\neq
i}^{K}p_{h}\int_{0}^{1}\mathbb{G(}h,G_{h}(G_{i}^{-1}(s))\times \nu
(G_{i}^{-1}(s))ds.  \label{ll03}
\end{equation}%
We are now going to compute the variance $\vartheta _{1}^{2}$ based on the
independent functional Browian bridges $\mathbb{G}(i,\cdot )$ which are limits of the
functional empirical process $\mathbb{G}_{n}(i,\cdot )$ \ respectively
associated with $\{G_{i}(Y_{i,j}),1\leq i\leq n_{i}\}$, $i=1,..,K.$
Straightforward calculations give what comes. First%
\begin{equation*}
A_{1}=\mathbb{E}C^{2}(1)=\sum_{i=1}^{K}p_{i}\mathbb{EG}%
^{2}(i,(g-g_{i})G_{i}^{-1}).
\end{equation*}%
We denote $l_{i}=(g-g_{i})G_{i}^{-1}$ in the sequel for sake of simplicity.
Next for%
\begin{equation*}
C(2)=\sum_{i=1}^{K}p_{i}^{1/2}\int_{0}^{1}\mathbb{G}(i,s)(p_{i}\nu -\nu
_{i})(G_{i}^{-1}(s))ds
\end{equation*}%
we have%
\begin{equation*}
A_{2}=\mathbb{E}(C^{2}(2))=\sum_{i=1}^{K}p_{i}\int_{0}^{1}\int_{0}^{1}(s%
\wedge t-st)c_{i}(t)c_{i}(s)dsdt
\end{equation*}%
\begin{equation*}
=\sum_{i}^{K}p_{i}\int_{0}^{G_{i}(Z)}\int_{0}^{G_{i}(Z)}(s\wedge t-st)(p_{i}%
\overline{\nu }-\overline{\nu }_{i})(G_{i}^{-1}(s))(p_{i}\overline{\nu }-%
\overline{\nu }_{i})(G_{i}^{-1}(t))dsdt,
\end{equation*}%
where $c_{i}(t)=(p_{i}\nu -\nu _{i})\left( G_{i}^{-1}(t)\right)$. Now for%
\begin{equation*}
C(3)=\sum_{i=1}^{K}\sqrt{p_{i}}{}\sum_{h\neq i}^{K}p_{h}\int_{0}^{1}\mathbb{%
G(}h,G_{h}(G_{i}^{-1}(s)))\times \nu (G_{i}^{-1}(s))ds,
\end{equation*}%
we have%
\begin{equation*}
A_{3}=\mathbb{E}(C^{2}(3))=\mathbb{E}\left\{ \sum_{i=1}^{K}p_{i}(\sum_{h\neq
i}^{K}K_{i,h})^{2}+\sum_{i\neq j}^{K}(p_{i}p_{j})^{1/2}(\sum_{h\neq
i}^{K}K_{i,h})(\sum_{h^{\prime }\neq j}^{K}K_{i,h^{\prime }})\right\} .
\end{equation*}%

\noindent Put
$$
K_{i,h}=p_{h}\int_{0}^{1}\mathbb{%
G(}h,G_{h}(G_{i}^{-1}(s)))\times \nu (G_{i}^{-1}(s))ds,
$$

split $A_{3}$ into

$$
A_{31} = \mathbb{E} \left(  \sum_{i=1}^{K} p_{i} ( \sum_{h\neq i}^{K} K_{i,h} )^{2} \right)
$$

\noindent and

\begin{equation*}
A_{32}=\mathbb{E}(\sum_{i=1}^{K}\sum_{i\neq j}^{K}\left( p_{i}p_{j}\right)
^{1/2}(\sum_{h\neq i}^{K}K_{i,h})(\sum_{h^{\prime }\neq j}^{K}K_{i,h^{\prime
}}))^{2}.
\end{equation*}

\noindent Now by using the independence of the centered stochastic process $\mathbb{G}(h, \cdots)$ for differents values of $h \in \{1,...,K\}$, one gets

$$
A_{31} =\mathbb{E} \left(   \sum_{i=1}^{K}p_{i}\sum_{h\neq i}^{K}K_{i,h}^{2} \right)
$$

\noindent and then

$$
A_{31} =\sum_{i=1}^{K}p_{i}\sum_{h\neq i}^{K}p_{h}^{2} \\
\int_{0}^{G_{i}(Z)}\int_{0}^{G_{i}(Z)} \left[ {G_{h}(G_{i}^{-1}(s))\wedge
G_{h}(G_{i}^{-1}(t))} \right.
$$

$$
\left. {-G_{h}(G_{i}^{-1}(s))G_{h}(G_{i}^{-1}(t))} \right] 
 \overline{\nu }(G_{i}^{-1}(s))\overline{\nu }(G_{i}^{-1}(t))dsdt.
$$

\noindent Next, one has

$$
A_{32}=\mathbb{E} \sum_{i=1}^{K}\sum_{j\neq i}^{K}(p_{i}p_{j})^{1/2}\sum_{h\neq i,h^{\prime
}\neq j}^{K}p_{h}p_{h^{\prime }} \int_{0}^{1}\int_{0}^{1}
$$

$$
 \mathbb{G(}h,G_{h}(G_{i}^{-1}(s)) 
\mathbb{G}(h^{\prime},G_{h^{\prime}}(G_{j}^{-1}(t)))  \nu(G_{i}^{-1}(s)) \nu(G_{j}^{-1}(t))dtds
$$

$$
=\sum_{i=1}^{K}p_{i}^{1/2}\sum_{j\neq i}^{K}p_{j}^{1/2}\sum_{h\notin
\{i,j\}}^{K}p_{h}^{2} \int_{0}^{G_{i}(Z)}\int_{0}^{G_{j}(Z)} \\
\left[ {G_{h}(G_{i}^{-1}(s))\wedge G_{h}(G_{j}^{-1}(t))} \right.
$$

$$
\left. {-G_{h}(G_{i}^{-1}(s))G_{h}(G_{j}^{-1}(t))} \right] \overline{\nu }(G_{i}^{-1}(s))\overline{\nu }(G_{j}^{-1}(t))dsdt,
$$

\noindent Now we have
\begin{equation*}
C(1)C(2)=\left( \sum_{i=1}^{K}p_{i}^{1/2}\mathbb{G}(i,\ell _{i})\right)
\left( \sum_{i=1}^{K}p_{i}^{1/2}\int_{0}^{1}\mathbb{G}(i,s)c_{i}(s)\text{ }%
ds\right)
\end{equation*}%

\begin{equation*}
=\sum_{i=1}^{K}\sum_{j=1}^{K}\left( p_{i}p_{j}\right) ^{1/2}\int_{0}^{1}%
\mathbb{G}(i,s)c(s)\mathbb{G}(j,\ell _{j})\text{ }c_{i}(s)\text{ }ds.
\end{equation*}%

\noindent and get 
\begin{equation*}
B_{1}=\mathbb{E}C(1)C(2)=\sum_{i=1}^{K}p_{i}\int_{0}^{1}\mathbb{E}(\mathbb{G}%
(i,s)\mathbb{G}(j,\ell _{i})\text{ }c_{i}(s)ds
\end{equation*}%
\begin{equation*}
=\sum_{i=1}^{K}p_{i}\int_{0}^{1}\left\{ \int_{-\infty
}^{G_{i}^{-1}(s)}(g-g_{i})(y)dG_{i}(y)-s\mathbb{E(}g-g_{i}\mathbb{)(}Y^{i}%
\mathbb{)}\right\} c_{i}(s)ds
\end{equation*}%

$$
=\sum_{i=1}^{K}p_{i} \int_{0}^{G_{i}(Z)}\left\{ {\int_{0}^{s\wedge G_{i}(Z)}(\overline{g}-%
\overline{g}_{i})(G_{i}^{-1}(t))dt} \right.
$$

$$
\left. {-s\int_{0}^{1}(\overline{g}-\overline{g}%
_{i})(G_{i}^{-1}(t))dt} \right\} (p_{i}\overline{\nu }-\overline{\nu }_{i})(G_{i}^{-1}(s))ds,
$$

We have next
$$
C(2)C(3)= \left( \sum_{i=1}^{K}p_{i}^{1/2}\int_{0}^{1}\mathbb{G}(i,s)c_{i}(s)ds\right)
$$

$$
\times \left( \sum_{i=1}^{K}p_{i}{}^{1/2}\sum_{h\neq i}^{K}p_{h}\int_{0}^{1}%
\mathbb{G(}h,G_{h}(G_{i}^{-1}(s))\times \nu (G_{i}^{-1}(s))ds\right)
$$

$$
=\sum_{i=1}^{K}\sum_{j=1}^{K}p_{i}^{1/2}p_{j}\sum_{i\neq j}^{K}p_{i} 
\int_{0}^{1}\int_{0}^{1}\mathbb{G(}i.s\mathbb{G}(i,G_{i}(G_{j}^{-1}(t)) c_{i}(s)\nu (G_{j}^{-1}(t))\mathbb{)}dsdt.
$$

\noindent It comes that%
$$
B_{2} = \mathbb{E}C(2)C(3)=\sum_{i=1}^{K}p_{i}^{3/2}\sum_{j\neq
i}^{K}p_{j}^{1/2} \int_{0}^{G_{i}(Z)}\int_{0}^{G_{j}(Z)}
$$

$$
[s\wedge
G_{i}(G_{j}^{-1}(t))-sG_{i}(G_{j}^{-1}(t))] 
\times (p_{i}\overline{\nu }-\overline{\nu }_{i})(G_{i}^{-1}(s))\overline{%
\nu }(G_{j}^{-1}(t))dsdt,
$$

\noindent Now finally for%

$$
C(1)C(3)=\left( \sum_{i=1}^{K}p_{i}^{1/2}\mathbb{G}(i,\ell _{i})\right) 
$$

$$
\times \left( \sum_{i=1}^{K}p_{i}^{1/2}{}\sum_{h\neq i}^{K}p_{h}\int_{0}^{1}%
\mathbb{G(}h,G_{h}(G_{i}^{-1}(s))\times \nu (G_{i}^{-1}(s))ds\right)
$$
\begin{equation*}
=\sum_{i=1}^{K}\sum_{j\neq i}^{K}p_{i}^{1/2}p_{j}^{1/2}\sum_{h\neq
i}^{K}p_{h}\int_{0}^{1}\mathbb{G(}h,G_{h}(G_{j}^{-1}(s))\mathbb{G}(i,\ell
_{i})\times \nu (G_{j}^{-1}(s))ds,
\end{equation*}%
where the $\ell_{i}'s$ are defined in \ref{def0}, we have%
\begin{equation*}
B_{3}=\mathbb{E}C(1)C(3)
\end{equation*}%
\begin{equation*}
=\sum_{i\neq j}^{K}p_{i}p_{j}\int_{0}^{1}\mathbb{E}\left\{ \mathbb{G}(i,\ell
_{i})\mathbb{G(}i,G_{i}(G_{j}^{-1}(s))\right\} \times \nu (G_{j}^{-1}(s))ds
\end{equation*}%

$$
\sum_{i=1}^{K}p_{i}^{3/2}\sum_{j\neq i}^{K}p_{j}^{1/2}
\int_{0}^{G_{j}(Z)} \left\{ {\int_{0}^{G_{i}(G_{j}^{-1}(s))\wedge G_{i}(Z)}(\overline{g}-%
\overline{g}_{i})(G_{i}^{-1}(t))dt} \right.
$$

$$
\left. {-G_{i}(G_{j}^{-1}(s))\int_{0}^{1}(%
\overline{g}-\overline{g}_{i})(G_{i}^{-1}(t))dt} \right\} 
\overline{\nu }(G_{j}^{-1}(s))ds.
$$

\noindent We have now finished the variance computation, that is 
\begin{equation*}
\vartheta _{1}^{2}=A_{1}+A_{2}+A_{3}+2(B_{1}+B_{2}+B_{3})
\end{equation*}
\end{proof}

$\bigskip $

\section{Conclusion}

\label{sec7}

\noindent We just illustrated how apply our results for the Sen Measure and
the Senegalese database ESAM I and the Mauritanian EPCV 2004 data. But It would be more interesting and
instructive to conduct large scale data-driven for the West African databases for example, for several measures. It would also be interesting to see the influence of the Kakwani parameter $k$ on the results. This study is
underway.

\section{Appendix} \label{sec8}
We would like to provide indications to the reader for using the techniques developped here. We have a zipped file at :
$$http://www/ufrsat.org/lerstad/sen-decomposabilite.rar$$
It includes the executable \textit{sendecomp.exe} file which performs the computation of $dg$. Here is how to proceed :\\
\begin{itemize}
\item[(i)] Download the zipped file and unzip him in a folder named, for instance, sen-decomposabilite.
\item[(ii)] Upload in the \textit{sen-decomposabilite} folder the following user files : The income file \textit{dep.txt} of size $n$ at most equal to $10 000$, the equivalent-adult file \textit{eq.txt} of the same size $n$ and finally the labels file \textit{labels.txt} including the names of the different strates. If the income file is already scaled for individuals, use an \textit{eq.txt} file of size $n$ having unity at each line. Le nomber of labels
is at most equal to $15$. They must be enumarated from to $1$ to $KK<16$.
\item[(iii)] Execute \textit{sendecomp.exe} by clicking on it. The user is prompted to provide the income file name, the equivalen-adult file name and the labels file name \textbf{without} the suffixs \textit{.txt}.
\item[(v)] The package provides the sen measures value for the differents strates and report the gap of decomposability value.
\item[(vi)] For the user's practice we provided in the zipped folder the following income variables (\textit{depm.txt}), equivalent-adult variable (\textit{eom.txt}) and labels (here areas) file named after \textit{regm.txt}.
\item[(vi)] If the data size exceeds $n=10000$ or the strates number exceeds $KK=15$, the user is free to write to the authors and adapted packages will be provided.
\end{itemize}

\bigskip
\noindent Finally for those who want to set their own packages in some langage, we provide a Visual Basic module including the main program and the subroutines.

\newpage

\end{document}